\documentclass[11pt]{article}

\usepackage{amssymb,amsthm,amsmath}
\usepackage{graphicx}
\usepackage[small]{caption}
\usepackage{subcaption}
\usepackage{epsfig}
	
\newcommand{\later}[1]{}
\newcommand{\old}[1]{}

\usepackage{amsfonts}

\usepackage[utf8]{inputenc}
\usepackage{fullpage}
\usepackage{framed}

\usepackage{enumerate}
\usepackage{url}
\usepackage{hyperref}

\newtheorem{theorem}{Theorem}
\newtheorem{lemma}[theorem]{Lemma}
\newtheorem{observation}[theorem]{Observation}
\newtheorem{proposition}[theorem]{Proposition}

\newtheorem{conjecture}[theorem]{Conjecture}

\newcommand{\etal}{{et~al.}}
\newcommand{\ie}{{i.e.}}

\newcommand{\area}{{\rm area}}

\newcommand{\dist}{{\rm dist}}
\newcommand{\diam}{{\rm diam}}

\newcommand{\NN}{\mathbb{N}} 
\newcommand{\RR}{\mathbb{R}} 
\newcommand{\eps}{\varepsilon}

\def\D{\mathcal D}
\def\R{\mathcal R}

\newcommand{\e}{\mathrm{e}}

\title{\sc Two trees are better than one \footnote{Research partially supported
    by the Hungarian Science Foundation (NKFIH) grant K-131529,
    the ERC grant no. 882971, "GeoScape," and by the Erd\H os Center.}}

\author{%
Adrian Dumitrescu\footnote{Algoresearch L.L.C., Milwaukee, WI 53217, USA,
  E-mail: \texttt{ad.dumitrescu@algoresearch.org}}
\and
J\'anos Pach\footnote{Alfr\'ed R\'enyi Institute of Mathematics,
Budapest, Hungary\@. Email: \texttt{pach@renyi.hu}}
\and
G\'eza  T\'oth\footnote{Alfr\'ed R\'enyi Institute of Mathematics,
  Budapest, Hungary\@. Email: \texttt{geza@renyi.hu}
}}

\begin{document}

\maketitle

\begin{abstract}
  We consider partitions of a point set into two parts,
  and the lengths of the minimum spanning trees of the original set and of the two parts.
  If $w(P)$ denotes the length of a minimum spanning tree of $P$,
  we show that every set $P$ of $n \geq 12$ points admits a bipartition
  $P= R \cup B$ for which the ratio $\frac{w(R)+w(B)}{w(P)}$ is strictly
  larger than $1$; and that $1$ is the largest number with this property.
  Furthermore, we provide a very fast algorithm that computes such a bipartition in
  $O(1)$ time and one that computes the corresponding ratio in $O(n \log{n})$ time.

In certain settings, a ratio larger than $1$ can be expected and sometimes guaranteed.
For example, if $P$ is a set of $n$ random points uniformly distributed in $[0,1]^2$
($n \to \infty$), then for any $\eps>0$, the above ratio in a maximizing partition
is at least $\sqrt2 -\eps$ with probability tending to $1$.
As another example, if $P$ is a set of $n$ points with spread at most $\alpha \sqrt{n}$, for
some constant $\alpha>0$, then the aforementioned ratio in a maximizing partition is
$1 + \Omega(\alpha^{-2})$.

All our results and techniques are extendable to higher dimensions.
\end{abstract}

\section{Introduction} \label{sec:intro}

A \emph{Euclidean minimum spanning tree} (EMST), for a set of $n$ points in the Euclidean plane
or Euclidean space is a spanning tree of the $n$ points that has the minimum total length
(the length of a tree edge is the Euclidean distance between its endpoints).
Write $w(P)$  for the length of an EMST of $P$.
Motivated by a question of Cultrera~\etal~\cite{Ed22, Ed23} arising from their work
on chromatic alpha shapes and persistence, we study the following problem.
For a finite set $P$ of $n$ points in the plane, and a bipartition $P=R \cup B$ into red and blue points,
consider the ratio
\begin{equation} \label{eq:ratio}
  \frac{w(R)+w(B)}{w(P)}.
\end{equation}
Call the maximum, over all non-trivial bipartitions of $P$ into two sets,
the \emph{max MST-ratio} of $P$, denoted $\gamma(P)$.
For $n=3,4,5$ we exhibit examples for which the ratios are at most $1$
and we have no larger examples with this property.
As such, we are tempted to conjecture that perhaps $6$ points suffice to guarantee
a bipartition whose MST-sum strictly exceeds the MST length of the set itself,
see Conjecture~\ref{conj:>1} below.
We come close to answering this question, see our Theorem~\ref{thm:general} below.

\begin{conjecture}  \label{conj:>1}
For any set $P$ of $n \geq 6$ points in the plane we have $\gamma(P) >1$.
\end{conjecture}

We first show using graph theoretic arguments that if $|P|=n \geq 3$,
the maximum ratio $\gamma(P)$ is at least $\frac{n-2}{n-1}$.
Note that this lower bound tends to $1$ as $n \to \infty$.
On the other hand, for any $n \geq 4$ and $\eps>0$, there exists a set of $n$ points
in the plane with $\gamma(P) \leq 1+\eps$. Therefore, the constant $1$~cannot be replaced by larger constant
in Conjecture~\ref{conj:>1}.
For $n=3$, there exists $P$ such that $\gamma(P)=1/2$, and
for $n=4$, there exists $P$ such that $\gamma(P) = \frac{\sqrt3 +1}{3} = 0.910\ldots$; see Fig. \ref{4points}.
\medskip

\begin{figure}[ht]
\begin{center}
\scalebox{0.45}{\includegraphics{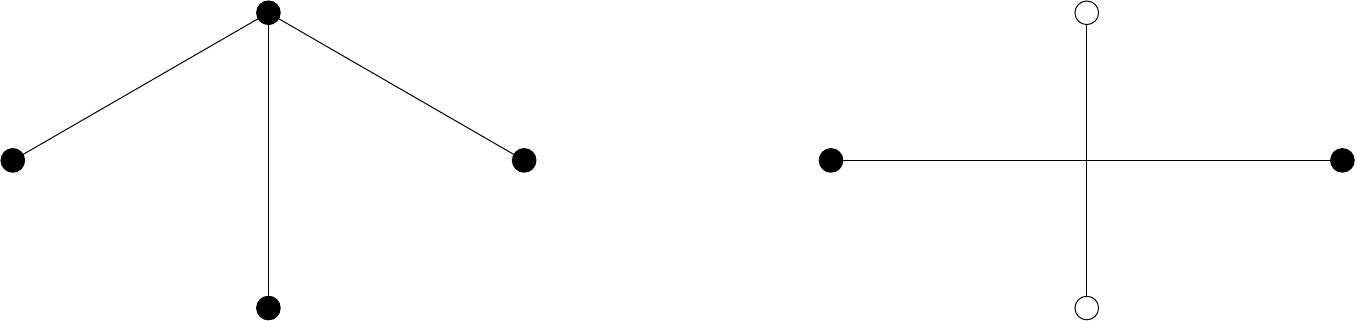}}
\caption{A point set $P$ with $\gamma(P) = \frac{\sqrt3 +1}{3} = 0.910\ldots$.}
\label{4points}
\end{center}
\end{figure}


\begin{proposition} \label{thm:lb-gamma}
  Let $P$ be any set of $n \geq 3$ points in the plane.

  Then $\gamma(P) \geq \frac{n-2}{n-1}$.
  On the other hand, for any $n \geq 3$ and $\eps>0$, there exists $P$ with $|P|=n$,
  such that $\gamma(P) \leq 1+\eps$.
\end{proposition}

The example showing the second statement was also found by Cultrera~\etal~\cite{Ed23}.
Using geometric arguments, we can improve the first statement and settle Conjecture~\ref{conj:>1}, for $n\ge 12$.

\begin{theorem}  \label{thm:general}
  For any set $P$ of $n \geq 12$ points in the plane, we have $\gamma(P) >1$.

  A suitable bipartition with a ratio $>1$ can be computed in $O(1)$ time, whereas
  the corresponding ratio can be computed in $O(n \log{n})$ time.
\end{theorem}

We also prove sharper lower bounds on $\gamma(P)$ in two cases:
(i) for random points uniformly distributed in
a square (or any other fat container region),  and (ii) for dense point sets.
In both cases we require the number of points to be sufficiently large.
Note that neither class of point-sets includes the other.
The proofs are simpler than that of Theorem~\ref{thm:general}.

We start with random points uniformly distributed.
According to a classic result of Beardwood, Halton, and Hammersley~\cite{BHH59},
the length of a minimum spanning tree of a random sample $\{ X_1,\ldots,X_n\}$ of $n$ points
uniformly distributed in the unit cube $[0,1]^k$ satisfies
\begin{equation} \label{eq:mst-random}
  {L(X_1,\ldots,X_n) / n^{1-1/k}} \to \beta(k)
\end{equation}
with probability one, where $\beta(k)>0$ is a constant depending on $k$.
Later, Bertsimas and Van Ryzin~\cite{BR90} proved that
${\beta(k) / \sqrt{k}} \to  {1 / \sqrt{2 \pi \e}}$, see also~\cite{St97}; however,
this more precise determination will not be needed here. Using only~\eqref{eq:mst-random}
we prove:

\begin{theorem} \label{thm:random}
Let $P$ be a set of $n$ random points uniformly distributed in $[0,1]^2$.

Then for any $\eps>0$, we have $\gamma(P) \geq \sqrt2 -\eps$ with probability tending to $1$
as $n \to \infty$.
\end{theorem}

We continue with dense point sets.
For a set $P$ of $n$ points in the plane, define
\[
\Delta(P)= \frac{\max\{\dist(a,b) : a,b \in P, a \neq b\}}{\min\{\dist(a,b) : a,b \in P, a \neq b\}}, \]
where $\dist(a,b)$ is the Euclidean distance between points $a$ and $b$.
The ratio $\Delta(P)$ is referred to as the \emph{aspect ratio} or the \emph{spread} of $P$;
see for instance~\cite{BLMN05}.
We assume without loss of generality that the minimum pairwise distance is $1$ and in this case
$\Delta(P)$ is the diameter $\diam(P)$ of $P$.
A~standard volume argument shows that if $P$ has $n$ points, then $\Delta(P) \geq c \, n^{1/2}$,
where $c>0$ is an absolute constant. If $n$ is large enough,
then $c \geq 2^{1/2} 3^{1/4} \pi^{-1/2} \approx 1.05$; see~\cite[Prop.~4.10]{Va92}.
An $n$-element  point set satisfying the condition $\Delta(P) \leq \alpha \, n^{1/2}$,
for some constant $\alpha>0$, is called here $\alpha$-\emph{dense}; see \cite{Va94}
or~\cite[Ch.~10]{BMP05}.

\begin{theorem} \label{thm:dense}
  Let $P$ be a dense set of $n$ points in the plane, with $\Delta(P) \leq \alpha n^{1/2}$, for some
  $\alpha \geq 2^{1/2} 3^{1/4} \pi^{-1/2}$.

  Then $\gamma(P) \geq 1 + \Omega(\alpha^{-2})$. More precisely, we have
  \[ \gamma(P) \geq 1 + \frac{1}{11 (2\alpha +1)^2},\;\;\; \text{ for large } n.\]
\end{theorem}

\paragraph{Related work.}
Graph decomposition is a classical topic in graph theory~\cite{Chung81,EH67,Lov66}.
In this note, we are interested in finding partitions of a points set into two parts such
that the total length of their MST's is as large as possible.

Estimating the length of a shortest tree or tour of $n$ points in the unit square with respect
to  Euclidean distances has been studied as early as 1940s and 1950s by Fejes T\'oth~\cite{Fej40},
Few~\cite{Few55}, and Verblunsky~\cite{Ve51}, respectively.
Clustering algorithms based on minimum and maximum spanning trees have been studied
by Asano~\etal~\cite{ABKY88}.

In the same spirit, recall that the well-know \emph{Euclidean Steiner ratio} is the infimum
over all finite point sets $P$ in the plane of the ratio between the lengths of a shortest Steiner tree
and that of minimum spanning tree of the set~\cite{GP68,CG85}.
Denoting this ratio by $\rho$, Gilbert and Pollak conjectured in 1968 that $\rho \geq \sqrt3/2=0.866\ldots$.
Early results due to Graham and Hwang~\cite{GH76} and Du and Hwang~\cite{DH83}
showed that $\rho \geq 1/\sqrt3 \approx  0.577$, and $\rho \geq 0.8$.
The current best lower bound on this ratio,  $\rho \geq \rho_0$, is due to Chung and Graham~\cite{CG85};
here $\rho_0 = 0.824168\ldots$ is the unique real root in the interval $(0.8,1)$ of a polynomial
of $12$th degree\footnote{$P(x) = x^{12} -4 x^{11} -2 x^{10} +40 x^9 - 31 x^8 - 72 x^7 + 116 x^6
+16 x^5 -151 x^4 + 80 x^3 + 56 x^2 -64 x +16$.}.

For every $n \geq 7$, there exists a set $P$ of $n$ points such that
$\gamma(P)$ can get arbitrarily close to  $\frac{2\sqrt3 +3}{3} \approx 2.154$. 
Using the above lower bound on the Steiner ratio, Cultrera et al.~\cite{Ed23} proved that
if $P$ is any set of $n \geq 3$ points in the plane,
then $\gamma(P) \leq 2/\rho_0 \approx 2.426$.

\section{Preliminaries} \label{sec:prelim}

We subsequently assume that $n \geq 3$, since otherwise the maximum MST-ratio is $0$.
We first observe the following upper bounds for $n=3,4,5$.

When $n=3$, the three vertices of an equilateral triangle $P$ give $\gamma(P) =1/2$.
Indeed, by symmetry the three possible bipartitions yield the same lengths, $1$,
whereas the MST length is $2$.
Proposition~\ref{thm:lb-gamma} shows that this bound is the best possible.

When $n=4$,  let $P$ be the four vertices of a rhombus of unit side-length and an angle of $60^\circ$.
The bipartition that gives the maximum ratio is one with $|R|=|B|=2$, where $R$ is a diametral pair,
thus  $\gamma(P)=  \frac{\sqrt3 +1}{3} \approx 0.910$.
It can be checked that if $P$ consists of the four vertices of a square, then the maximum ratio is
slightly larger, $\gamma(P)=  \frac{2 \sqrt2}{3} \approx 0.942$.
Proposition~\ref{thm:lb-gamma} yields that $\gamma(P) \geq 2/3$, for $|P|=4$.

When $n=5$, let $P$ consist of the following five points (written as complex numbers):
$ 0, 1, \omega, \omega^2, \omega^3$, where $\omega = \cos \frac{\pi}{3} + i \sin  \frac{\pi}{3}$
is the $6$th root of unity. It is easy to check that the bipartition that gives the maximum ratio
is one with $R=\{1,\omega^3\}=\{-1,1\}$, and so $\gamma(P)=\frac{2+2}{4}=1$.
Proposition~\ref{thm:lb-gamma} yields that $\gamma(P) \geq 3/4$, for any point set $P$ with $|P|=5$.

\paragraph{Disjoint disks.}
Let $P$ be a finite set of points in $\RR^d$. Let $\D(P) =\{\omega(p) \colon p \in P\}$,
where $\omega(p)$ is the largest closed ball centered at $p$ and containing no other points
from $P$ in its interior. Given $p \in P$, let $r(p)$ denote the radius of $\omega(p)$.

A set of points $P$ in $\RR^d$ is said to satisfy the \emph{disjoint balls} property
if $\D(P)$ contains two disjoint balls. For the plane, we use the term \emph{disjoint disks} property.
A key (easy) observation about this property is its monotonicity.

\begin{observation} \label{obs:monotone}
  Let $P$ be a set of points in the plane that satisfies the disjoint disks property.
  If $P \subset Q$, then $Q$ also satisfies this property.
\end{observation}
\begin{proof}
  Assume that $\omega(p)$ and $\omega(p')$ are disjoint, where $p,p' \in P$.
  Adding the points in $Q \setminus P$ maintains this property, since
  the disks centered at $p$ and $p'$ can only get smaller.
\end{proof}

\begin{lemma} \label{lem:disjoint-disks}
  Let $P$ be a set of points in the plane that satisfies the disjoint disks property.

  Then there exist two points $p,q \in P$ such that $|pq| > r(p) + r(q)$.
\end{lemma}
\begin{proof}
  Let $p,q \in P$ be the centers of two disjoint disks in $\D(P)$.
  Then $p,q$ clearly satisfy the requirement.
\end{proof}

According to a result of K{\'e}zdy and Kubicki~\cite{KK97},
any set of $n \geq 12$ points in the plane satisfies the
\emph{disjoint disks} property; we record this fact in the following lemma.
On the other hand, this property cannot be guaranteed for any set of $8$ points;
and it is conjectured that it holds for any set of $n \geq 9$ points~\cite{FL94,HJLM93}.

\begin{lemma} \label{lem:12}  {\rm (K{\'e}zdy and Kubicki~\cite{KK97})}
Let $P$ be a set of $n$ points in the plane with $n \geq 12$.

Then $P$ satisfies the disjoint disks property.
\end{lemma}

\paragraph{Proof of Proposition~\ref{thm:lb-gamma}.}

\emph{Upper bound.} Take $p_1=(0,0)$, $p_n=(1,0)$, and the rest of points in the interval $(0,\eps/n)$.
  It is easily verified that $w=1$, and $w_1 + w_2 \leq 1 + \eps$: indeed, the point $p_n$
  appears only in one tree, whereas the length of the second tree is less than $\eps$, which yields
  \[ \gamma(P) =\frac{w_1 + w_2}{w} \leq 1+\eps, \]
  as required.

  \medskip
  \emph{Lower bound.} Let $T_0$ be an MST of $P$, with $w=|T_0|$ and
  \[ \gamma(P) = \frac{w_1 + w_2}{w}, \]
  where $w_1 + w_2$, where $w_1$ and $w_2$ are the lengths of the MSTs, say, $T_1$ and $T_2$,
  for the red and the blue points in a bipartition of \emph{maximum} length.
 We claim that $w_1 + w_2 \geq \frac{n-2}{n-1} \, w$.
 Indeed, consider the bipartition $P= R \cup B$ obtained by deleting the shortest edge, say, $e$, of $T_0$.
 Note that $T_0 \setminus \{e\}$ consists of two spanning trees $T_R$ and $T_B$, of $R$ and $B$, respectively.
 Moreover, $T_R$ and $T_B$ are MSTs of $R$ and $B$, respectively.
 Indeed, assume for contradiction that
 there is spanning tree $T'_R$ of $R$ with a smaller length than $T_R$, \ie, $|T'_R| < |T_R|$.
 Then $T'_R \cup \{e\} \cup T_B$ is a spanning tree of $P$ of length
 \[ |T'_R|  + |e|  + |T_B| < |T_R|  + |e|  + |T_B| = w, \]
 a contradiction. Similarly, there is no spanning tree $T'_B$ of $B$ with a smaller length than $T_B$.

 Since $T$ has $n-1$ edges, we have $|e| \leq w/(n-1)$.
 The length of the resulting bipartition in terms of spanning trees is at least
 \[ w - \frac{w}{n-1} = \frac{n-2}{n-1} \, w. \]
 Since $w_1 + w_2$ is the length of a maximum bipartition (in the same terms), we have
 \[ w_1 + w_2 \geq \frac{n-2}{n-1} \, w, \]
as claimed.  Consequently,
   \[ \gamma(P) =  \frac{w_1 + w_2}{w} \geq \frac{n-2}{n-1}, \]
   as required.
\qed

\subsection{Key lemmas}

The following lemmas can be formulated for general graphs; however, here we restrict ourselves
to the simpler geometric setting.

\begin{lemma} \label{lem:mst-delete-2}
  Let $T$ be a minimum spanning tree of a complete geometric graph $G=K_n(P)$, and $p,q \in P$
  be the centers of two disjoint disks in $\D(P)$.

  Then $w(P \setminus \{p,q\}) \geq w -r(p) - r(q)$, where $w =w(P)$.
\end{lemma}
\begin{proof}
  Assume for contradiction that $w(P \setminus \{p,q\}) < w -r(p)-r(q)$, and let $T'$ be
  a MST of $P \setminus \{p,q\}$. Add $p,q$ to $T'$ by connecting them via edges of length
  $r(p)$ and $r(q)$, respectively. We obtain a spanning tree of $G$ whose length is strictly less than
  \[ w(P \setminus \{p,q\}) + r(p)+r(q) < w -r(p)-r(q)  + r(p)+r(q) =w, \]
 a contradiction.
\end{proof}

We can formulate this statement in a slightly more general form that will be used in the proof
of Theorem~\ref{thm:dense}.

\begin{lemma} \label{lem:mst-delete-k}
  Let $T$ be a minimum spanning tree of a complete geometric graph $G=K_n(P)$,
  where $w =w(P)$ is the length of $T$.
  Let $p_1q_1, p_2q_2, \ldots,p_k q_k$ be a set of $k$ point pairs on $2k$ distinct points in $P$.

  Then $w(P \setminus \{p_1,p_2,\ldots,p_k\}) \geq w - \sum_{i=1}^k |p_i q_i|$.
\end{lemma}
\begin{proof}
Assume for contradiction that
$w(P \setminus \{p_1,p_2,\ldots,p_k\}) < w - \sum_{i=1}^k |p_i q_i|$, and let $T'$ be
  a MST of $P \setminus \{p_1,p_2,\ldots,p_k\}$.
Add $p_1,p_2,\ldots,p_k$ to $T'$ by connecting them via the $k$ edges
$p_1q_1, p_2q_2, \ldots,p_k q_k$, respectively; note that $q_1,q_2,\ldots,q_k$ are
vertices in $T'$.
 We obtain a spanning tree of $G$ whose  length is strictly less than
\[ w - \sum_{i=1}^k |p_i q_i| + \sum_{i=1}^k |p_i q_i| =w, \]
 a contradiction.
\end{proof}

\section{The general case}

In this section we prove Theorem~\ref{thm:general}.
Let $p,q \in P$ be the centers of two disjoint disks whose existence is guaranteed by
Lemma~\ref{lem:12}. Consider the bipartition of $P$:
\[ R = P \setminus \{p,q\}, \ \ B=\{p,q\}. \]
Let $T$ be a minimum spanning tree of $P$, where $w =w(P)$.
By using Lemma~\ref{lem:mst-delete-2}, we have
\[ w(R) \geq w -r(p) - r(q),  \text{ and } w(B) = |pq|. \]
Consequently, be Lemma~\ref{lem:disjoint-disks} we obtain
\[ \gamma(P) \geq  \frac{w(R) + w(B)}{w} \geq \frac{w - r(p) - r(q) + |pq|}{w} > 1, \]
as required. See Fig. \ref{fig:bipartition}

\medskip

\begin{figure}[ht]
\begin{center}
\scalebox{0.45}{\includegraphics{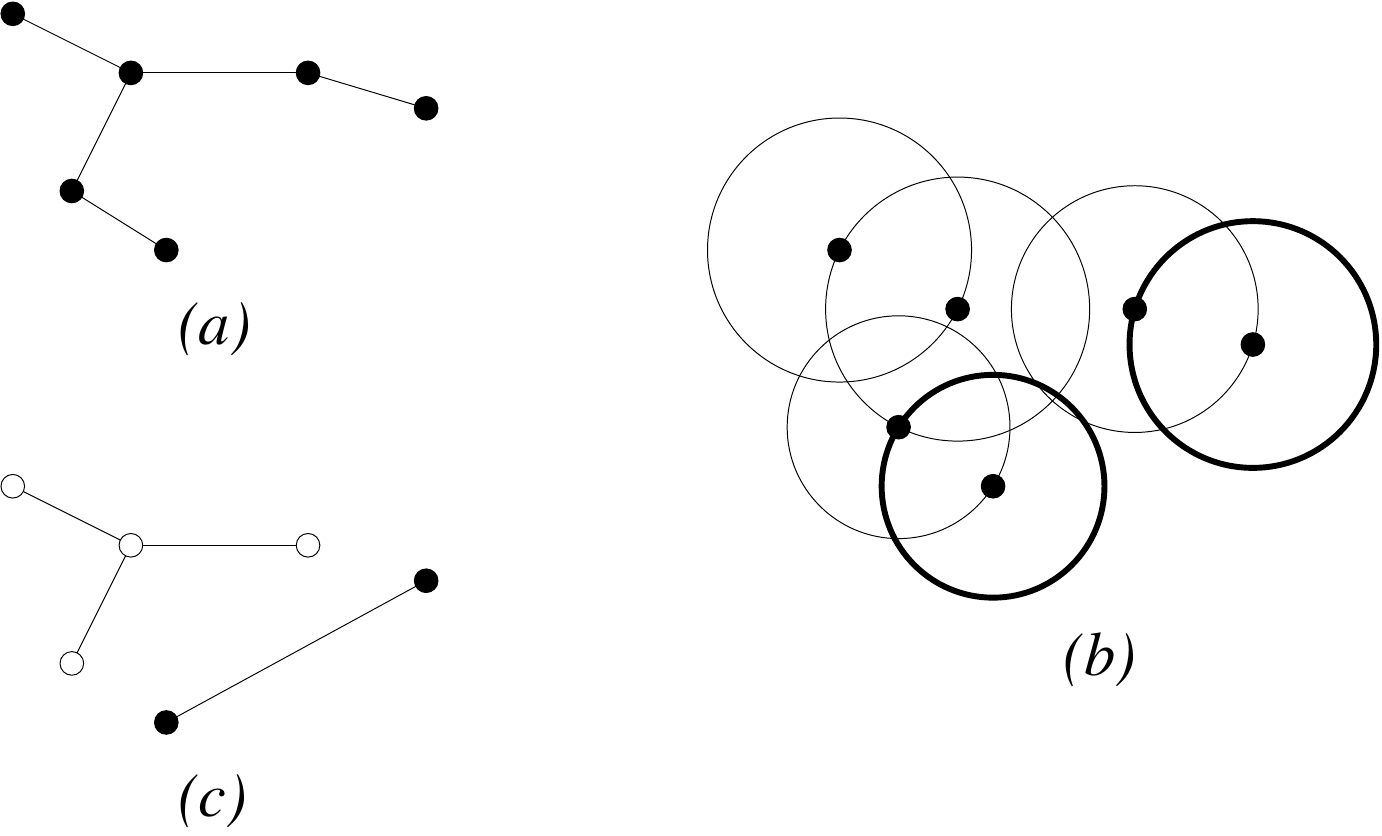}}
\caption{(a) A point set $P$
  with a minimum spanning tree, (b) a pair of  disjoint disks in $\D(P)$ defines a bipartition,
  (c) the minimum spanning trees of the two parts.}
\label{fig:bipartition}
\end{center}
\end{figure}

\vspace{-0.5cm}

\paragraph{Algorithm.}
Given $P$, where $|P|=n$, the algorithm only keeps the first (or any) $12$ points of the input
for consideration. Denoting this subset by $P'$, the centers of two disjoint balls,
say, $p$ and $q$, are found in $P'$ in $O(1)$ time.
The corresponding partition is $R = P \setminus \{p,q\}$, and $B=\{p,q\}$, and the algorithm outputs
the smaller set $B$ in $O(1)$ time. Its correctness follows from Observation~\ref{obs:monotone}.

The three relevant spanning trees and the corresponding ratio can be computed in $O(n \log{n})$ time,
see, eg, \cite{BCK+08}. This completes the proof of Theorem~\ref{thm:general}.
\qed

\medskip
Interestingly enough, any sufficiently large point set determines many pairs of disjoint disks and,
therefore, many bipartitions with a ratio $>1$.

\begin{proposition} \label{thm:many}
  Let $P$ be a set of $n\ge 12$ points in the plane.
  Then 
there are at least $n(n-1)/132$ pairs of disjoint disks in $\D(P)$.
\end{proposition}
\begin{proof}
By Lemma~\ref{lem:12}, every $12$ points determine at least one pair of disjoint disks in $\D(P)$.
Since each fixed pair appears in exactly ${n-2 \choose 10}$ $12$-tuples, $\D(P)$ contains at least
\[ {n \choose 12} \Big{/} {n-2 \choose 10} = \frac{n(n-1)}{132} \]
pairs of disjoint disks.
\end{proof}

\smallskip

By a more
careful calculation, the constant $1/132$
in the above statement can be improved.

\section{Special cases} \label{sec:special}

\subsection{Random point sets} \label{subsec:random}

\paragraph{Proof of Theorem~\ref{thm:random}.}
Let $\beta=\beta(2)$ be the constant from Equation~\eqref{eq:mst-random}.
Let $P=a_1,a_2,\ldots,a_n$ be a sequence of $n$ random points uniformly distributed in $[0,1]^2$,
where $n$ is a large even number. Consider the bipartition $P= R \cup B$ consisting of the
first $n/2$ and the last $n/2$ points, \ie, $R=\{a_1,\ldots,a_{n/2}\}$, and $B=\{a_{n/2+1},\ldots,a_n\}$.

By~\eqref{eq:mst-random},
\[ w=w(P) \sim \beta \sqrt{n}, \ \ w(R) \sim \beta \sqrt{n/2},
\text{ and } w(B) \sim \beta \sqrt{n/2}, \]
with probability one. Consequently, as $n \to \infty$,
\[ \gamma(P) \geq \frac{w(R) + w(B)}{w} \geq
\frac{2\beta \sqrt{n/2} - o(\sqrt{n})}{\beta \sqrt{n} +o(\sqrt{n})}
\geq \sqrt2 -\eps, \]
with probability tending to $1$, as claimed.
\qed

\subsection{Dense point sets} \label{subsec:dense}

\paragraph{Proof of Theorem~\ref{thm:dense}.}
Let $P$ be an $n$-element $\alpha$-dense set and $T$ be a minimum spanning tree of $P$ of length $w$.
Since $\diam(P) \leq \alpha \sqrt{n}$, we may assume that $P$ is contained in an axis-aligned square $Q$
of side-length $\alpha \sqrt{n}$.
Subdivide $Q$ into $n/4$ axis-parallel squares, called \emph{cells}, of side-length $2\alpha$;
indeed, $\alpha^2 n/(4 \alpha^2)=n/4$.
We may assume without loss of generality that no point in $P$ lies on a cell boundary.
We need two preliminary lemmas.

\begin{lemma} \label{lem:cell-ub}
Each cell contains at most $\frac{2}{\sqrt3} (2\alpha +1)^2$ points in $P$.
\end{lemma}
\begin{proof}
Let $\sigma \subset Q$ be any cell.
Let $\sigma'$ be the axis-aligned square concentric with $\sigma$ and whose side-length is a unit larger
than that of $\sigma$. Its side-length is $2\alpha + 1$.
Obviously,  $\sigma'$ contains all open disks of radius $1/2$ centered at the points in $P \cap \sigma$.
Moreover, since $P$ is $\alpha$-dense, these $n$ disks are disjoint.

Obviously $\sigma'$ is a so-called \emph{tiling domain}, \ie, a domain that can be used to
tile the whole plane~\cite[Ch.~3.4]{FFK23}.
Let $m$ denote the number of points in $P \cap \sigma$. A standard packing argument
requires that $m \frac{\pi}{4} \leq \frac{\pi}{\sqrt{12}} \, \area(\sigma')$,
and yields $m \leq \frac{2}{\sqrt3} (2\alpha +1)^2$.
\end{proof}

Let $\Sigma$ be the set of $n/4$ cells in $Q$.
A cell $\sigma \in \Sigma$ is said to be \emph{rich} if it contains at least $2$ points in $P$,
and \emph{poor} otherwise. Let $\R \subset \Sigma$ denote the set of rich cells; and write $r=|\R|$.

\begin{lemma} \label{lem:rich-lb}
  There are at least $\frac{3 \sqrt3}{8(2\alpha +1)^2} \, n$ rich cells; that is,
  $r \geq \frac{3 \sqrt3}{8(2\alpha +1)^2} \, n$.
\end{lemma}
\begin{proof}
 Assume for contradiction that $r < \frac{3 \sqrt3}{8(2\alpha +1)^2} \, n$.
By Lemma~\ref{lem:cell-ub} the total number of points in rich cells is at most
\[ r \cdot \frac{2}{\sqrt3} (2\alpha +1)^2 <
\frac{3 \sqrt3}{8(2\alpha +1)^2} \, n \cdot  \frac{2}{\sqrt3} (2\alpha +1)^2 = \frac34 \, n. \]
The total number of points in poor cells is less than $1 \cdot \frac{n}{4} = \frac{n}{4}$.
Thus the total number of points in $P$ is strictly less than $n$, a contradiction
that completes the proof.
\end{proof}

The following upper bound on $w$ will be needed at the end of the proof.

\begin{lemma} \label{lem:Q-lb}
We have $w \leq 1.4 \alpha n$, for $n$ sufficiently large.
\end{lemma}
\begin{proof}
Few~\cite{Few55} proved that the (Euclidean) length of a shortest cycle (tour)
through $n$ points in the unit square $[0,1]^2$ is at most $\sqrt{2n}+7/4$, and
the same upper bound holds for the minimum spanning tree.
A~slightly better upper bound $1.392 \sqrt{n} + 7/4$, is due to Karloff~\cite{Ka89}.
For large $n$, this is at most $1.4 \sqrt{n}$. Scaling by $\alpha \sqrt{n}$ immediately
yields that $w \leq 1.4 \alpha n$.
\end{proof}

We continue with the proof of the theorem.
Decompose $\R$ into $25$ subsets of rich cells as in the decomposition of a $5 \times 5$ square
into unit squares. More precisely, consider the tiling of the plane by $5 \times 5$ axis-aligned tiles
of side-length $2\alpha$. Arbitrarily fix one of the tiles.
A subset corresponds to the rich cells in $Q$ having the same position
in the fixed tile. Each subset satisfies the following properties
(recall the side-length of each cell $2\alpha$):
\begin{enumerate} [i] \itemsep 2pt
  \item the distance between any two points in the same cell is at most $ 2\sqrt2 \alpha$.
  \item the distance between any two points in different cells is at least $8 \alpha$.
\end{enumerate}

By averaging, one of the subsets, $\R'$, contains $k \geq \lceil |\R|/25 \rceil$ rich cells. Note that
\begin{equation} \label{eq:k}
  k \geq \frac{3\sqrt3}{200 (2\alpha +1)^2} \, n,
  \end{equation}
by Lemma~\ref{lem:rich-lb}.
Arbitrarily choose two points $p_i,q_i$ in each of the $k$ rich cells.
Now consider the bipartition of $P$:
\[ R = P \setminus \{p_1,p_2,\ldots,p_k\}, \ \ B=\{p_1,p_2,\ldots,p_k\}. \]

By Lemma~\ref{lem:mst-delete-k} and the first property of $\R'$ we have
\[ w(R) \geq w - \sum_{i=1}^k |p_i q_i| \geq w - 2 \sqrt2 \alpha k. \]
On the other hand, by the second property of $\R'$ we have
\[ w(B) \geq 8 \alpha (k-1). \]
Consequently, by Lemma~\ref{lem:Q-lb} and the fact that $n$ (thus also $k$ according to~\eqref{eq:k})
is sufficiently large, we obtain
\begin{align*}
\gamma(P) &\geq  \frac{w(R) + w(B)}{w} \geq \frac{w + 2(4-\sqrt2) \alpha k  -8 \alpha}{w}
\geq \frac{w + 5 \alpha k}{w} \\
&\geq 1 + \frac{5 \alpha k}{1.4 \alpha n}
= 1 +\frac{3 \sqrt3}{56 (2\alpha +1)^2} \geq 1 + \frac{1}{11 (2\alpha +1)^2}\\
&= 1 + \Omega(\alpha^{-2}),
\end{align*}
as required. This completes the proof of Theorem~\ref{thm:dense}.
\qed

\medskip
It is likely that using a hexagonal tiling gives a slightly better constant in the bound,
however, the dependence of the bound on $\alpha$ remains quadratic.

\section{Concluding remarks}

The techniques used to handle the planar case immediately extend to $\RR^d$ for any fixed dimension $d$.
Specifically we need the \emph{disjoint balls} property; we rely on estimates
due to F{\"u}redi and Loeb~\cite{FL94}, given in~\eqref{eq:kappa}. Let $\kappa(d)$ denote the maximum
number of balls in $\RR^d$ that can form a configuration where each ball intersects every other ball
but does not contain the center of any other ball. As mentioned earlier, it is known that
$8 \leq \kappa(2) \leq 11$. The following bounds hold in $d$-space:
\begin{equation} \label{eq:kappa}
  1.25^d \leq \kappa(d) \leq (1.887\ldots + o(1))^d.
\end{equation}

The proof of Theorem~\ref{thm:d-general} is analogous to that of Theorem~\ref{thm:general} and
is therefore omitted.

\begin{theorem}  \label{thm:d-general}
  For every fixed $d$, there exists $n_d \in \NN$ with the following property.
  For any set $P$ of $n \geq n_d$ points in $\RR^d$ we have $\gamma(P) >1$.
  A suitable bipartition with this ratio can be computed in $O(1)$ time.
\end{theorem}

We have seen that the disjoint balls property implies that the maximum MST ratio is
strictly larger than $1$. However, the two properties are not equivalent.

Another related question is how the $\gamma$ ratio changes when points are added to
a set. Suppose that $P_0$ is a point set with $\gamma(P_0)>1$, and an infinite sequence of
points, $p_1,p_2,\ldots$, are successively added to $P_0$. By Theorem~\ref{thm:general},
$\gamma(P_0 \cup \{p_1,\ldots,p_n\}) > 1$ for every $n$ such that $|P_0|+n \geq 12$.
Is it possible that $\gamma(P_0 \cup \{p_1,\ldots,p_n\})$ tends to $1$ as $n \to \infty$?
We give a positive answer. Assuming that the first $n$ points $p_1,\ldots,p_n$
are known, let $P_n = P_0 \cup \{p_1,\ldots,p_n\}$, $w_n=w(P_n)$, and $D_n= \diam(P_n)$,
where $a_n,b_n$ is a diameter pair in $P_n$. Assume that $\gamma(P_n) \leq 1+\eps$ for some $\eps>0$,
we can select  $p_{n+1}$ so that $\gamma(P_{n+1}) \leq 1+\eps/2$, for every $n$.
Let $p_{n+1}$ be a point on the line incident to $a_n,b_n$ at a sufficiently large distance
$f(w_n,\eps)$ from $b$, with $a_n,b_n,p_{n+1}$ appearing in this order.
The verification that the above inequality holds, for a suitable function $f$,
is left to the reader. Observe that the above construction is analogous to the upper bound
construction in Proposition~\ref{thm:lb-gamma}; and works in any dimension.

\smallskip
We conclude with three open problems:

\begin{enumerate} \itemsep 2pt

\item We have seen that $\gamma(P)>1$ can be arbitrarily close to $1$ even for large $n$,
  whereas it is separated from $1$ by a positive gap for randomly uniformly distributed sets
  and for $\alpha$-dense sets. In what other cases can one deduce $\gamma(P)>1 +\lambda$
  for a positive constant $\lambda$?

\item Is there an efficient algorithm for computing the maximum ratio, $\gamma(P)$?

\item Is the $\sqrt2$ lower bound for randomly uniformly distributed sets
  in the unit square tight?

\end{enumerate}

\noindent\textbf{Acknowledgement.} We are grateful to Herbert Edelsbrunner for calling our attention to the subject
and for his valuable remarks. In particular, he communicated to us the construction in the second half of Proposition~\ref{thm:lb-gamma}.


\begin{thebibliography}{99}

\bibitem{ABKY88}
  Tetsuo Asano,  Binay Bhattacharya, Mark Keil, and  Frances Yao,
  Clustering algorithms based on minimum and maximum spanning trees,
\emph{Proc. $4$th Annual Symposium on Computational Geometry},
1988, pp.~252--257.

\bibitem{BLMN05}
Yair Bartal, Nathan Linial, Manor Mendel, and Assaf Naor,
On metric Ramsey-type phenomena,
\emph{Ann. Math.}
\textbf{162(2)} (2005), 643--709.

\old{
\bibitem{BE51}
  Paul Bateman and  Paul Erd{\H{o}}s,
Geometrical extrema suggested by a lemma of Besicovitch,
\emph{The American Mathematical Monthly}
\textbf{58(5)} (1951), 306--314.
} 

\bibitem{BHH59}
  Jillian Beardwood,  John H. Halton, and John M. Hammersley,
  The shortest path through many points,
  \emph{Mathematical Proceedings of the Cambridge Philosophical Society}
\textbf{55(4)} (1959), 299--327.

\bibitem{BCK+08}
  Mark~de~Berg, Otfried Cheong, Marc van Kreveld, and Mark~Overmars,
\emph{Computational Geometry}, 3rd edition, Springer, Heidelberg, 2008.

\bibitem{BR90}
Dimitris J. Bertsimas and Garrett Van Ryzin,
An asymptotic determination of the minimum spanning tree and minimum matching constants
in geometrical probability,
\emph{Operations Research Letters}
\textbf{9(4)} (1990), 223--231.

\bibitem{BMP05}
Peter Bra\ss , William Moser, and J\'anos Pach,
\emph{Research Problems in Discrete Geometry},
Springer, New York, 2005.

\bibitem{CG85}
Fan R. K. Chung and Ronald L. Graham,
  A new bound for Euclidean Steiner minimal trees,
  \emph{Annals of the New York Academy of Sciences}
\textbf{440(1)}, (1985), 328--346.

\bibitem{Chung81}
  Fan R. K. Chung,
On the decomposition of graphs,
\emph{SIAM Journal on Algebraic Discrete Methods}
\textbf{2(1)}, (1981), 1--12.

\bibitem{Ed22}
Sebastiano Cultrera di Montesano, Ond\v rej Draganov, Herbert Edelsbrunner, and Morteza Saghafian. Persistent homology of chromatic alpha complexes, \emph{arXiv:2212.03128}.

\bibitem {Ed23}
Sebastiano Cultrera di Montesano, Ond\v rej Draganov, Herbert Edelsbrunner, and Morteza Saghafian.
Measuring the mingling of colored points,
communication at the \emph{Convex and Discrete Geometry Workshop},  Sept. 4--8, 2023,
Alfr\'ed R\'enyi Institute of Mathematics, Budapest, Hungary.

\bibitem {DH83}
Ding -Z. Du and Frank K. Hwang,
A new bound for the Steiner ratio,
\emph{Transactions of the American Mathematical Society}
\textbf{278(1)}, (1983), 137--148.

\bibitem {EH67}
Paul Erd\H{o}s and Andr{\'a}s Hajnal,
On decomposition of graphs,
\emph{Acta Math. Acad. Sci. Hungar}
\textbf{18} (1967), 359--377.

\bibitem{Fej40}
L\'aszl\'o Fejes T\'oth,
  {\"U}ber einen geometrischen Satz,
  \emph{Mathematische Zeitschrift}
\textbf{46(1)} (1940), 83--85.

\bibitem{FFK23}
L\'aszl\'o Fejes T\'oth, G\'abor Fejes T\'oth, and W{\l}odzimierz Kuperberg,
\emph{Lagerungen},
Springer Nature Switzerland, 2023.

\bibitem{Few55} Leonard Few,
The shortest path and shortest road through $n$ points,
\emph{Mathematika}
\textbf{2} (1955), 141--144.

\bibitem{FL94}
Zolt{\'a}n F{\"u}redi and  Peter A. Loeb,
On the best constant for the Besicovitch covering theorem,
\emph{Proceedings of the American Mathematical Society}
\textbf{121(4)} (1994), 1063--1073.

\bibitem{GP68}
Edgar N. Gilbert and Henry O. Pollak,
Steiner minimal trees,
\emph{SIAM Journal on Applied Mathematics}
\textbf{16(1)} (1968), 1--29.

\bibitem{GH76}
Ronald L. Graham and Frank K. Hwang,
Remarks on Steiner minimal trees,
\emph{Bull. Inst. Math. Acad. Sinica}
\textbf{4(1)} (1976), 177--182.

\bibitem{HJLM93}
Frank Harary,  Michael S. Jacobson,  Marc J. Lipman, and Fred R.  McMorris,
Abstract sphere-of-influence graphs,
\emph{Mathematical and Computer Modelling},
\textbf{17(11)} (1993), 77--83.

\bibitem{Ka89} Howard J. Karloff,
How long can a Euclidean traveling salesman tour be?
\emph{SIAM Journal of Discrete Mathematics}
\textbf{2} (1989), 91--99.

\bibitem {KK97}
A. E. K{\'e}zdy and G.  Kubicki,
$K_{12}$  is not a closed sphere-of-influence graph,
\emph{Intuitive Geometry}
J{\'a}nos Bolyai Mathematical Society, Budapest,
\textbf{6} (1997), 383--397.

\bibitem{Lov66}
L{\'a}szl{\'o} Lov{\'a}sz,
On decompositions of graphs,
\emph{Studia Scientiarum Mathematicarum Hungarica}
\textbf{1} (1966), 237–-238.

\bibitem {St97}
  Michael J. Steele,
\emph{Probability Theory and Combinatorial Optimization},
SIAM, Philadelphia, 1997.

\bibitem{Va92} Pavel Valtr,
Convex independent sets and 7-holes in restricted planar point sets,
\emph{Discrete \& Computational Geometry}
\textbf{7(2)} (1992), 135--152.

\bibitem{Va94} Pavel Valtr,
\emph{Planar point sets with bounded ratios of distances},
PhD Thesis, Freie Universit\"at Berlin (1994).

\bibitem {Ve51}
Samuel Verblunsky,
On the shortest path through a number of points,
\emph{Proceedings of the American Mathematical Society}
\textbf{2(6)} (1951), 904--913.

\end{thebibliography}
\end{document}